\documentclass{amsart}   
\usepackage[utf8]{inputenc}
\usepackage{algorithm}
\usepackage{algorithmic}
\usepackage{amsthm}     
\usepackage[plainpages=false,pdfpagelabels,colorlinks=true,citecolor=blue,hypertexnames=false]{hyperref}  

\def\K{\mathbb K}
\def\Q{\mathbb Q}

\def\N{\mathbb N}
\def\M{\mathsf M}
\def\Int{I}
\renewcommand{\algorithmicrequire} {\textbf{\textsc{Input:}}}
\renewcommand{\algorithmicensure}  {\textbf{\textsc{Output:}}}
\newcommand{\lclm}{\operatorname{lclm}} 
\newcommand{\gcld}{\operatorname{gcld}} 
\newtheorem{theo}{Theorem}
\newtheorem{lemma}{Lemma}
\newtheorem{prop}{Proposition}
\newtheorem{corollary}{Corollary}
\theoremstyle{definition}
\newtheorem{example}{Example}

\begin{document}

\title{Chebyshev Expansions for\\ Solutions of Linear Differential Equations}

\address{Algorithms Project, Inria Paris-Rocquencourt}
\author{Alexandre Benoit}
\author{Bruno Salvy}
\email{Alexandre.Benoit@inria.fr,Bruno.Salvy@inria.fr}

\thanks{This work was supported in part by the Inria-Microsoft Research Joint Centre.}
\thanks{To appear in the proceedings of ISSAC'09.}
       
\begin{abstract}
A Chebyshev expansion is a series in the basis of Chebyshev polynomials of the first kind. When such a series solves a linear differential equation, its coefficients satisfy a linear recurrence equation. We interpret this equation as the numerator of a fraction of linear recurrence operators. This interpretation lets us give a simple view of previous algorithms, analyze their complexity, and design a faster one for large orders.
\end{abstract}

\maketitle

\section{Introduction}
Chebyshev series are series of the form
\begin{equation}\label{chebseries}
f(x)=\frac{c_0}{2}+\sum_{n=1}^{\infty}c_nT_n(x),
\end{equation}
where $T_n$ denotes the $n$th Chebyshev polynomial of the first kind. These polynomials can be defined by
\begin{equation}\label{defTn}
T_n(\cos\theta)=\cos(n\theta),
\end{equation}
so that these series behave like Fourier series. Thus in particular, this series converges pointwise to~$f$ on~$[-1,1]$ if $f$ is continuous there, while the convergence is uniform if~$f$ satisfies a Dini-Lipschitz condition or is of bounded variation (and \emph{a fortiori} if it is differentiable), see, e.g.~\cite{GilSeguraTemme2007a,MasonHandscomb2003}.
Then truncations of the series provide polynomials with good approximation properties on the interval $[-1,1]$, which makes these series an interesting data structure for real functions~\cite{Trefethen2007}. 

Orthogonality of the~$T_n$ leads to the following integral representation of the coefficients:
\[c_n=\frac{2}{\pi}\int_{-1}^1{\frac{f(x)T_n(x)}{\sqrt{1-x^2}}dx}\qquad (n\in \N).\]
We say that~$f$ \emph{admits a Chebyshev expansion} $\sum^\star{c_nT_n}$ when these integrals converge, the symbol~$\sum^\star$ accounting for the factor~$1/2$ in front of~$c_0$ in~\eqref{chebseries}.

In the frequent case when~$f$ is a solution to a linear differential equation, Clenshaw~\cite{Clenshaw1957} has given a numerical scheme to compute the coefficients~$c_n$ without computing all these integrals.
In that case, the coefficients~$c_n$ obey a linear recurrence equation. A method for the computation of this recurrence has been showed by several authors, first for small orders~\cite{FoxParker1968,Luke1969a}, then in more generality by Paszkowski~\cite{Paszkowski1975} and in the context of (early) symbolic computation by Geddes~\cite{Geddes1977}. We call this method ``Paszkowski's algorithm''.
The  use of this recurrence to compute the coefficients numerically is discussed in~\cite{Wimp1984}.
Paszkowski's method has been further improved by Lewanowicz~\cite{Lewanowicz1976} who gave an algorithm computing a smaller order recurrence in some cases. However, Lewanowicz's algorithm is not much discussed in the literature since it looks complicated (see the original article and the comment in \cite[p.~186]{Wimp1984}).
More recently, other methods have been given by Rebillard~\cite{Rebillard1998} and Rebillard and Zakraj\v{s}ek~\cite{RebillardZakrajsek2006}. 

In this work, we give a simple unified presentation of most of these algorithms, and design a faster one for large orders. Postponing the proofs and rigorous definitions, the basic idea can be presented by analogy with the computation of a recurrence for coefficients of Taylor series. The monomial basis~$M_n(x)=x^n$
satisfies
\begin{equation}\label{taylor}
xM_n(x)=M_{n+1}(x),\qquad M_n'(x)=nM_{n-1}(x).
\end{equation}
The analogous relations on the Chebyshev polynomials are easily derived from~\eqref{defTn} and trigonometry: 
\begin{align}
	2xT_n(x)&=T_{n+1}(x)+T_{n-1}(x),\label{eq:Cheb1}  \\
	2(1-x^2)T'_n(x)&=-nT_{n+1}(x)+nT_{n-1}(x).\label{eq:Cheb2}
\end{align}
Given a series $f(x)=\sum{c_nM_n(x)}$, \eqref{taylor} leads to expressions for the coefficient of $M_n(x)$ in~$xf$ and $f'$: multiplication by~$x$ maps to a negative shift on the indices; differentiation maps to a positive shift of the index followed by multiplication by~$n+1$. Algebraically, we thus get an algebra morphism mapping~$x$ to~$X$ and $d/dx$ to~$D$, with $X:=S^{-1}$, $D:=(n+1)S$. Here, $S$ denotes the shift operator: $u(n)\mapsto u(n+1)$, that does not commute with multiplication by~$n$.
Now, if $f$ is solution of a linear differential equation 
\[
p_k(x)f^{(k)}(x)+\dots+p_0(x)f(x)=0,
\]
we deduce a recurrence operator $p_k(X)D^k+\dots+p_0(X)$ for its Taylor coefficients.
\begin{example} The simplest example is the exponential, for which $f'-f=0$ translates into~$D-1=(n+1)S-1$ ($1$ denotes identity), which gives the recurrence $(n+1)c_{n+1}-c_n=0$ satisfied by~$c_n=1/n!$. 
\end{example}
The procedure for a series~$f(x)=\sum{c_nT_n(x)}$ starts similarly: multiplication by~$x$ maps to
\begin{equation}\label{imagex}
	X:=(S+S^{-1})/2.
\end{equation}
The difference comes from the factor~$(1-x^2)$ in~\eqref{eq:Cheb2}. The operation of differentiation followed by multiplication by~$(1-x^2)$ is readily seen to map to~$(S-S^{-1})n/2$, but no simple linear operation for the Chebyshev coefficients of $f'$ exists. The idea at this stage is to divide by $1-x^2$ afterwards by introducing a \emph{formal inverse} of~$1-X^2$. Thus we write~$D:=(1-X^2)^{-1}(S-S^{-1})n/2$. This can be further simplified since
$1-X^2=-(S-S^{-1})^2/4$,
so that 
\begin{equation}\label{imagedx}
D:=2(S^{-1}-S)^{-1}n.
\end{equation}
We call such an expression a \emph{fraction of recurrence operators}.
\begin{example}\label{example exp} For the exponential, we now get
	\[D-1=2(S^{-1}-S)^{-1}n-1=(S^{-1}-S)^{-1}(2n-(S^{-1}-S)).\] 
The last term is an analogue of reduction to the same denominator. The final factor will be called the numerator of the fraction. It corresponds to the recurrence
\begin{equation}\label{recIn}
2nc_n-c_{n-1}+c_{n+1}=0.
\end{equation}
It turns out that in this example, the Chebyshev coefficients are known:
$c_n=2I_n(1)$, where~$I_n$ is the modified Bessel function of the first kind, and they do satisfy~\eqref{recIn}.
\end{example}              
This example generalizes. We show here that all the algorithms mentioned above can be interpreted as first rewriting the input linear differential equation in one way or another, then applying the morphism above and finally returning the numerator of the result. In the case of Lewanowicz's algorithm the fraction is normalized (its numerator and denominator are relatively prime), which is why its output may have smaller order.

In Section~\ref{sec:Ore}, we give the formal setting for fractions of recurrence operators, together with the basic algorithms. This is then applied to the specific case of Chebyshev series in Section~\ref{sec:Chebyshev}. Then we give a compact presentation of Paszkowski's and Rebillard's algorithms, provide a complexity analysis and design a faster algorithm in Section~\ref{sec:algo}. We briefly comment on the different approach taken by Rebillard and Zakraj\v{s}ek in~\S\ref{rebillard-zakrajsek}. We conclude in Section~\ref{sec:examples} with a few examples.

\section{Fractions of Recurrence Operators}\label{sec:Ore}
We use  Ore's framework of non-commutative polynomials~\cite{Ore33}, that we now recall.
                                                                      
\subsection{Ore Polynomials}
The rings of linear differential operators and of linear recurrence operators are special cases of rings of Ore polynomials. They possess the commutation rules
\[\frac{d}{dx}p(x)=p(x)\frac{d}{dx}+p'(x);\qquad
Sp(n)=p(n+1)S.\]
More generally, a ring of polynomials in an indeterminate~$\partial$ with coefficients in a field~$\K$ is an \emph{Ore polynomial ring} when its product is defined by associativity from
\begin{equation}\label{commut}
\partial p=\sigma(p)\partial+\delta(p),\qquad p\in\K
\end{equation}
where for all $a$ and $b$ in~$\K$, 
\begin{alignat*}{3}
\sigma(a+b)&=\sigma(a)+\sigma(b),\qquad &\sigma(ab)&=\sigma(a)\sigma(b),\qquad \sigma(a^{-1})=\sigma(a)^{-1},\\
\delta(a+b)&=\delta(a)+\delta(b),&\delta(ab)&=\sigma(a)\delta(b)+\delta(a)b.
\end{alignat*}
The ring is denoted~$\K\langle\partial;\sigma,\delta\rangle$. Linear differential operators are obtained with~$\sigma=\operatorname{Id}$ and~$\delta=d/dx$; linear recurrences operators with~$\sigma=S$ and~$\delta=0$. 

The main property of these rings is that the degree (with respect to~$\partial$) of a product is the sum of the degrees of its factors. (In particular, there are no zero-divisors). From there, it is not difficult to write an algorithm for Euclidean division on the right. Once right Euclidean division is available, the Euclidean algorithm and its extended version follow and can be used to compute: greatest common right divisors, denoted gcrd; least common left  multiples, denoted lclm; the corresponding cofactors for the Bézout identity and for the lclm~\cite{Ore33,BronsteinPetkovsek1996}.

When~$\sigma$ is invertible, we also get Euclidean division on the left, and from there greatest common left divisors (gcld), least common right multiples (lcrm) and the corresponding cofactors by the Euclidean algorithm.
If moreover $\delta=0$, as is the case for recurrence operators, it is also possible to define Laurent polynomials with~$\partial\partial^{-1}=\partial^{-1}\partial=1$ and $\partial^{-1}a=\sigma^{-1}(a)\partial^{-1}$. These are denoted~$\K\langle\partial,\partial^{-1};\sigma\rangle$. 

The rings we use in this work are the ring of linear differential operators denoted $\Q(x)\langle\partial_x;\operatorname{Id},d/dx\rangle$ and the ring of linear recurrence operators $\Q(n)\langle S,S^{-1};S\rangle$ (with a different meaning for both~$S$).

Apart from their non-commutativity, Ore polynomials generally behave like ordinary polynomials. A notable exception is divisibility. 
\begin{example}The recurrence operator $P=(n+1)^{-1}(S+1)$ is relatively prime with~$Q=nS+n+2$, but $Q$ is a right divisor of $P^2$.
\end{example}
Still, the following property holds (and similarly for gcrd's when they exist):
\begin{equation}\label{gcld}
\gcld(AB,AC)=A\gcld(B,C).
\end{equation}
Indeed, $A$ is a left divisor of~$\gcld(AB,AC)$ and the remaining factor has to be a left divisor of both~$B$ and~$C$. The converse divisibility is clear.

In order to distinguish the action of an operator from the product in these rings of operators, which corresponds to composition of actions, we use the notation~$\cdot$ for the former. Thus~$\partial_x\cdot f=f'$, $S\cdot u_n=u_{n+1}$.

\subsection{Fractions}

Ore's construction of fractions parallels the commutative case.
Given two non-zero polynomials~$Q_1$ and~$Q_2$, by definition of the lclm, there exist two polynomials~$\tilde{Q}_1$ and~$\tilde{Q}_2$ such that
\[\lclm(Q_1,Q_2)=\tilde{Q}_1Q_1=\tilde{Q}_2Q_2.\]
With this notation, the pairs~$(P_1,Q_1)$ and~$(P_2,Q_2)$ are called equivalent when
$\tilde{Q}_1P_1=\tilde{Q}_2P_2$.
This can be verified to be an equivalence relation and the class is called a fraction and denoted~$Q_1^{-1}P_1$ (which is equal to $Q_2^{-1}P_2$). This construction makes the set of fractions a non-commutative field. 

Reduction to the same denominator for sums is given by
\begin{equation}\label{eq:redsum}
Q_1^{-1}P_1+Q_2^{-1}P_2=\lclm({Q}_1,Q_2)^{-1}(\tilde{Q}_1 P_1+\tilde{Q}_2 P_2),
\end{equation}
as can be checked by left multiplication with~$\lclm(Q_1,Q_2)$. 

To compute the reduction of a product of two fractions $Q_1^{-1}P_1,Q_2^{-1}P_2$, the starting point is the lclm of~$Q_2$ and the numerator~$P_1$. There exist two polynomials~$\hat{P_1}$ and~$\hat{P_2}$ such that
\[\lclm(Q_2,P_1)=\hat{P}_1P_1=\hat{Q}_2Q_2.\]
Then, $Q_1^{-1}P_1=(\hat{P}_1Q_1)^{-1}\hat{P}_1P_1$ and $Q_2^{-1}P_2=(\hat{Q}_2Q_2)^{-1}\hat{Q}_2P_2$, so that finally
\begin{equation}\label{eq:redprod}
Q_1^{-1}P_1Q_2^{-1}P_2=(\hat{P}_1Q_1)^{-1}{\hat{Q}_2 P_2}.
\end{equation}

\subsection{Irreducible Fractions}
Having in mind our use of fractions for recurrence operators, we now concentrate on the case when~$\sigma$ is invertible, so that gcld's are available. The results here are probably known, but we did not find them in the literature.

A fraction~$Q^{-1}P$ is called \emph{irreducible} when $\gcld(P,Q)=1$.
\begin{prop}
Assume $\sigma$ is invertible and let~$A^{-1}B$ be a fraction. Then there exists an irreducible fraction equal to $A^{-1}B$. Moreover, its numerator and denominator are unique up to a factor in~$\K$.
\end{prop}
\begin{proof}
Existence follows from dividing out numerator and denominator by~$\gcld(A,B)$. 
Assume $Q_1^{-1}P_1=Q_2^{-1}P_2$ and  $\gcld(P_1,Q_1)=\gcld(P_2,Q_2)=1$. 
By definition of equivalence, $\tilde{Q}_1P_1=\tilde{Q}_2P_2$, where $\tilde{Q}_1Q_1=\tilde{Q}_2Q_2=\lclm(Q_1,Q_2)$. Moreover this lclm relation implies $\gcld(\tilde{Q}_1,\tilde{Q}_2)=1$.
Now,
\begin{multline*}
\tilde{Q}_1=\tilde{Q}_1\gcld(P_1,Q_1)=\gcld(\tilde{Q}_1P_1,\tilde{Q}_1Q_1)\\
=\gcld(\tilde{Q}_2P_2,\tilde{Q}_2Q_2)=\tilde{Q}_2\gcld(P_2,Q_2)=\tilde{Q}_2,
\end{multline*}
where we use~\eqref{gcld}. But since $\gcld(\tilde{Q}_1,\tilde{Q}_2)=1$, necessarily $\tilde{Q}_1=\tilde{Q}_2=1$ and then
$P_1=P_2$ and $Q_1=Q_2$.
\end{proof}

The following lemma is useful in the computation of recurrences for Chebyshev series.
\begin{lemma} \label{lemirred}
Assume~$\sigma$ is invertible and let $Q_2^{-1}P_2$ be an irreducible fraction and $P_1$ a polynomial. Then $\hat{P}_1^{-1}\hat{Q}_2P_2$ with $\hat{Q}_2Q_2=\hat{P}_1P_1=\lclm(Q_2,P_1)$ is irreducible and equal to $P_1Q_2^{-1}P_2$.
\end{lemma}
\begin{proof}
We have  $\gcld(\hat{Q}_2,\hat{P}_1)=1$ by definition of the $\lclm$.
The polynomial $g=\gcld(\hat{P}_1,\hat{Q}_2P_2)$ is a left divisor of $\hat{P_1}{P_1}=\hat{Q_2}Q_2$, and therefore is a left divisor of \[\gcld(\hat{Q}_2P_2,\hat{Q}_2Q_2)=\hat{Q}_2\gcld(P_2,Q_2)=\hat{Q}_2.\] Thus $g$ is a left divisor of both $\hat{P}_1$ and $\hat{Q}_2$, hence is~$1$.
\end{proof}

\section{Recurrences for Chebyshev Coefficients}\label{sec:Chebyshev}
We now have the theoretical tools to prove that a morphism from linear differential operators to linear recurrence operators produces fractions whose numerators give recurrences for the coefficients of Chebyshev series solutions. 

The algorithms then become easy to state, their algorithmic difficulty being concentrated in the Euclidean algorithm in the previous section.

\subsection{Morphism}
We define a morphism of~$\Q$-algebras from $\Q[x]\langle\partial_x;\operatorname{Id},d/dx\rangle$ to the field of fractions of $Q(n)\langle S,S^{-1};S\rangle$ by
\[\varphi(x)=X:=\frac12({S+S^{-1}}),\quad\varphi(\partial_x)=D:=(S^{-1}-S)^{-1}(2n).\]
The proof that $\varphi$ is a well-defined morphism of non-commutative rings reduces to checking the commutation $\varphi(\partial_x x)=\varphi(x\partial_x+1)$. Indeed,
\begin{align*}
XD+1&=\frac12({S+S^{-1}})(S^{-1}-S)^{-1}(2n)+1\\
&=(S^{-1}-S)^{-1}({S+S^{-1}})n+1\\
&=(S^{-1}-S)^{-1}\Bigl(\bigl((n+1)S+(n-1)S^{-1}\bigr)+(S^{-1}-S)\Bigr)\\
&=DX.
\end{align*}

\begin{algorithm}[t]
	\caption{Lewanowicz' algorithm\label{algo:Lewanowicz}}
	\begin{algorithmic}
		\REQUIRE $L:=\sum_{i=0}^k p_i(x)\partial^i$
		\ENSURE $(P,Q)$ such that $\varphi(L)=Q^{-1}P$
		\STATE $P:=p_k(X)$
		\STATE $Q:=1$
		\FORALL {$i$ from $k-1$ to $0$}
			\STATE Compute $\lclm((S^{-1}-S),P)=\hat{P}P=\hat{U}(S^{-1}-S)$.
			\STATE $Q:=\hat{P}Q$
			\STATE $P:=\hat{U}2n+Qp_i(X)$
		\ENDFOR
		\RETURN $(P,Q)$
	\end{algorithmic}
\end{algorithm}
\subsection{Horner's Rule and Lewanowicz' Algorithm}
\begin{prop}\label{horner}
	Let~$L=p_k(x)\partial_x^k+\dots+\partial_xp_0(x)$ be a linear differential operator in $\Q[x]\langle\partial_x;\operatorname{Id},d/dx\rangle$. The evaluation of $\varphi(L)$ by Horner's rule
	\[\varphi(L)=(\dotsb(p_k(X)D+p_{k-1}(X))D+\dotsb)D+p_0(X)\]
using Eqs.~\eqref{eq:redsum} and~\eqref{eq:redprod} for the computation of sums and products produces a fraction~$Q^{-1}P$ that is irreducible.
\end{prop}
The algorithm deduced from this statement (Algorithm~\ref{algo:Lewanowicz}) is due to Lewanowicz. It is made very clear by the use of fractions of recurrence operators. The proof that the numerator of its output gives a recurrence for the Chebyshev coefficients is given in the next section.
\begin{proof}
We prove that each iteration of the loop produces~$(P,Q)$ that are relatively prime and such that 
\begin{equation}\label{M_i}
	Q^{-1}P=:M_i=\varphi(p_k(x)\partial^{k-i}+\dots+p_i).
\end{equation}
Initially, $i=k$ and $\varphi(p_k(x))$ is a polynomial, so that $Q=1$ and the property holds.
If it holds for~$M_i$, the next stage of the loop computes $Q^{-1}PD+p_{i-1}(X)$.
Recall that~$D=(S^{-1}-S)^{-1}(2n)$. Then let $\lclm((S^{-1}-S),P)=\hat{P}P=\hat{U}(S^{-1}-S)$. Lemma~\ref{lemirred} applied to the inverse $(S^{-1}-S)P^{-1}Q$ implies that~$\gcld(\hat{P}Q,\hat{U})=1$. It follows that~$\gcld(\hat{P}Q,\hat{U}+\hat{P}Qp_{i-1}(X)/(2n))=1$. Again by Lemma~\ref{lemirred} applied to the inverse, multiplying by~$2n$ on the right preserves irreducibility and the property holds for~$M_{i-1}$.
\end{proof}
We quote without proof the following result.
\begin{prop}[Lewanowicz]\label{prop:Lewanowicz}
When the leading coefficient $p_k(x)$ of the differential equation does not vanish at~$1$ or~$-1$, then all the gcrd's are trivial, $Q=D^{-i}$ at step~$i$ and the resulting~$Q$ is~$D^{-k}$.
\end{prop}
This is related to the fact that~$4(X^2-1)=(S^{-1}-S)^2$.
\subsection{Chebyshev Expansions}
\subsubsection{Main Theorem}
We now prove our main result: the morphism defined above behaves as expected with respect to Chebyshev expansions.
\begin{theo}\label{main}
Let $L=p_0(x)+\dots+p_k(x)\partial_x^k$ be a linear differential operator of order~$k$ with polynomial coefficients. Let $f\in\mathcal{C}^k(\left]-1,1\right[)$
be such that either of the following hypotheses holds:
\begin{align*}
\tag{H}\int_{-1}^1{\frac{f^{(k)}(x)}{\sqrt{1-x^2}}\,dx}&\text{ is convergent;}\\
\tag{H'}\int_{-1}^1{\frac{(1-x^2)^kf^{(k)}(x)}{\sqrt{1-x^2}}\,dx}&\text{\ is convergent
and $(1-x^2)^i|p_i$, $i=0,\dots,k$}.
\end{align*}
Then $f$ admits a Chebyshev expansion $\sum^\star{u_nT_n}$,
$L\cdot f$ admits a Chebyshev expansion $\sum^\star{v_nT_n}$ and the sequences $(u_n)$ and $(v_n)$ are related by $P\cdot u_n=Q\cdot v_n$, for any~$(P,Q)$ such that $Q^{-1}P=\varphi(L)$. In particular, if~$L\cdot f=0$, then the Chebyshev coefficients of~$f$ satisfy $P\cdot u_n=0$ for any numerator of~$\varphi(L)$.
\end{theo}
The easy case is when (H) holds. Hypothesis (H') makes it possible to deal with some functions that are singular at $\pm1$, but whose singularities are not ``too bad'': they are regular singular points. 
\begin{proof}
First, convergence of the integral in (H) or (H') implies convergence of the analogous integral where $f^{(k)}$ is replaced by~$f^{(i)}$ for~$i=k-1,\dots,0$ as well as the integrals where these functions are multiplied by~$T_n(x)$, $n\in\N$. This shows that both~$f$ and~$L\cdot f$ admit Chebyshev expansions.

If the result holds for any numerator of~$\varphi(L)$ then in particular it has to hold for the numerator of its irreducible form.
Conversely, if~$P\cdot u_n=Q\cdot v_n$, then $RP\cdot u_n=RQ\cdot v_n$ for any~$R$, so that it is also sufficient to prove the result for an irreducible~$\phi(L)$.

\begin{lemma}[Basic Cases] Under the same hypotheses, the result holds for $L$ a constant times identity, $L=x$, $L=\partial_x$ if (H) holds, $L=(1-x^2)\partial_x$ if (H') holds.
\end{lemma}
\begin{proof}
If $L=\lambda$ is a constant times identity then $P=\lambda$, $Q=1$ and $v_n=\lambda u_n$ clearly holds.

If $L=x$,  Eq.~\eqref{eq:Cheb1} implies
\[
v_n=\frac{2}{\pi}\int_{-1}^{1}{\frac{f(x)xT_n(x)}{\sqrt{1-x^2}}\,dx}
=\frac{2}{\pi}\int_{-1}^{1}{\frac{f(x)(T_{n+1}(x)+T_{n-1}(x))}{2\sqrt{1-x^2}}\,dx}
=X\cdot u_n.
\]

If $L=\partial_x$ and (H) holds, 
we use the following variant of Eq.~\eqref{eq:Cheb2} when $n\neq0$
\[\frac{T_n}{\sqrt{1-x^2}}=\left(\frac{T_{n+1}-T_{n-1}}{2n\sqrt{1-x^2}}\right)'=-\frac{1}{n^2}(\sqrt{1-x^2}T_n')'\]
that can be checked from~\eqref{defTn}. The continuity of $f'$ and the convergence of the integral in (H) imply that integrating by parts is possible and this gives
\begin{align*}
u_n&=\frac{2}{\pi}\int_{-1}^{1}{\frac{f(x)T_n(x)}{\sqrt{1-x^2}}\,dx}\\
   &=\left[-\frac{2f(x)\sqrt{1-x^2}T_n'}{{\pi n^2}}\right]_{-1}^1
\frac{2}{\pi}\int_{-1}^1{\frac{f'(x)(T_{n-1}(x)-T_{n+1}(x))}{2n\sqrt{1-x^2}}\,dx}
=D^{-1}\cdot v_n.
\end{align*}
Both limits of the term between brackets are~0, by convergence of the integral~$u_n$.

The case when $n=0$ reduces to checking $v_{-1}=v_1$, that does not depend on~$f$.

If $L=(1-x^2)\partial_x$ and (H') holds, we start from
\[(-2\sqrt{1-x^2}T_n(x))' = \frac{(n+1)T_{n+1}(x)-(n-1)T_{n-1}(x)}{\sqrt{1-x^2}}.\]
An argument similar to the previous one then gives
\begin{align*}
	(n+1)u_{n+1}-(n-1)u_{n-1}
	&=\frac{2}{\pi}\int_{-1}^1{f(x) \frac{(n+1)T_{n+1}(x)-(n-1)T_{n-1}(x)}{\sqrt{1-x^2}} dx}\\
	      &=2 \frac{2}{\pi}\int_{-1}^1{(1-x^2)f'(x)\frac{T_n}{\sqrt{1-x^2}} dx}= 2 v_n,
\end{align*}
which proves the result since~$\phi((1-x^2)\partial_x)=(1-X^2)D=(S-S^{-1})n/2$.
\end{proof}
\begin{lemma}[Product]\label{product}
Assume the result holds for~$L_2$ with~$f$, as well as for another operator~$L_1$ with $L_2\cdot f$.
Let~$\varphi(L_1)=Q_1^{-1}P_1$ and~$\varphi(L_2)=Q_2^{-1}P_2$, these fractions being irreducible.
Let $\lclm(Q_2,P_1)=\hat{P}_1P_1=\hat{Q}_2Q_2$, and assume $(\hat{P}_1Q_1)^{-1}\hat{Q}_2P_2$ is irreducible. Then the result holds for~$L_1L_2$ with~$f$.
\end{lemma}
\begin{proof}
Let $v_n,w_n,u_n$ be related by $P_1\cdot v_n=Q_1\cdot w_n$, $P_2\cdot u_n=Q_2\cdot v_n$. Then
\[\hat{P}_1Q_1\cdot w_n=\hat{P}_1P_1\cdot v_n=\hat{Q}_2Q_2\cdot v_n=\hat{Q}_2P_2\cdot u_n,\]
whence the result.
\end{proof}
As a consequence, the result holds when~$L=\lambda x^i$ is a monomial, by induction.



\begin{lemma}[Sum]\label{sum}Assume the result holds for an operator~$L$ with $f$ and for a polynomial~$p$ with the same~$f$. Then it holds for $L+p$ with $f$.
\end{lemma}
\begin{proof}
Let $\phi(L)=Q^{-1}P$ be irreducible.
If~$P\cdot u_n=Q\cdot v_n$, $w_n=p(X)\cdot u_n$, then
\[Q\cdot(v_n+w_n)=(P+Qp(X))\cdot u_n.\] 
This proves the property for $L+p$ since $\gcld(Q,P+Qp(X))=\gcld(Q,P)=1$. 
\end{proof}
The result now holds for~$L$ an arbitrary polynomial, as a sum of its monomials.

Let finally $\theta=\partial_x$ if (H) holds and $\theta=(1-x^2)\partial_x$ if (H') does. In both cases, $L$ can be written
$q_k(x)\theta^k+\dots+q_0(x)$ with polynomial $q_k,\dots,q_0$. The hypothesis on~$f$ implies that the result holds for~$L=1$ with $\theta^i\cdot f$ for $i=0,\dots,k$ and therefore also for $L=q_i(x)$ with $\theta^i\cdot f$ by Lemma~\ref{product}.

Let $L_{k}=q_k(x)$ and $L_{i}=L_{i+1}\theta+q_{i}$ for $i=k-1,\dots,0$. Let 
$\phi(L_i)=Q_i^{-1}P_i$ be irreducible.
We prove by induction that the result holds for $L_i$ with $\theta^if$.
For $i=k$, the result has just been proved.
If the result holds for~$L_{i+1}$ with $\theta^{i+1}f$, then we obtain an irreducible $Q_{i+1}^{-1}P_{i+1}\theta$: when $\theta=\partial_x$ this follows from Lemma~\ref{product}, while when~$\theta=(1-x^2)\partial_x$, $L_{i+1}$ itself is a polynomial (by induction). Thus the result holds for~$L_{i+1}\theta$ with $\theta^if$. Since it also holds for $q_i(x)$ with $\theta^if$ and $q_i(x)$ is a polynomial, we get the result for their sum by Lemma~\ref{sum}. Thus by induction the result holds for~$L_0$ with~$f$, which concludes the proof of Theorem~\ref{main}.
\end{proof}

\subsubsection{Examples}
\begin{example}
The function $\exp(x)$ satisfies (H) for any~$k$. This proves the recurrence~\eqref{recIn} computed in Example~\ref{example exp}.
\end{example}
\begin{example}\label{ex-5}
The function $(1-x^2)^{-1/4}$ is annihilated by~$2(1-x^2)\partial_x-x$. Hypothesis (H) does not hold, but (H') does. Application of the morphism gives $P=(2n+3)S^2-(2n+1)$, $Q=-2S$, so that the theorem asserts that the Chebyshev coefficients satisfy 
\begin{equation}\label{eq coeffs ex-5}(2n+3)c_{n+2}=(2n+1)c_n.
\end{equation} The actual values can be computed by standard properties of the Beta integrals and indeed
\[c_n=\begin{cases} 
	0 &\text{if $n$ is odd,} \\ \frac{2\Gamma\left(\frac{n}{2}+\frac{1}{4}\right)}{\sqrt{\pi}\Gamma\left(\frac{n}{2}+\frac{3}{4}\right)} & \text{otherwise.}
\end{cases}\]
\end{example}
\begin{example}\label{ex arccos} The function $\arccos x$ gives an example showing that analytic hypothesis such as (H) or (H') are necessary. This function is annihilated by $L=(1-x^2)\partial_x^2-x\partial_x$. Direct application of the morphism gives $P=n^2$, $Q=1$, which would suggest that the recurrence is $n^2c_n=0$. However, neither (H) nor (H') holds in this case. Left multiplying $L$ by $(1-x^2)$ gives a new operator such that (H') holds. Then the theorem proves that the coefficients are annihilated by
\begin{equation}\label{eq arccos}
(n+4)^2S^4-2(n+2)^2+n^2.
\end{equation}
This can be checked against the actual coefficients:
\[c_n=\begin{cases} 
	\pi&\text{if $n=0$},\\
	0 &\text{if $n>0$ is even,} \\
	-\frac{4}{n^2\pi} & \text{otherwise.}
\end{cases}\]
\end{example}

\section{Algorithms}\label{sec:algo}
We now cast the algorithms of 
Paszkowski \cite{Paszkowski1975} and Rebillard~\cite{Rebillard1998} as computations of the numerator of a fraction of recurrence operators. We also propose a new faster algorithm. All three algorithms compute the same recurrence. Starting from
\begin{equation}\label{eq:L}
L=\sum_{i=0}^k p_i(x)\partial_x^i,
\end{equation}
they avoid the need for fractions by replacing differentiations by integrations, exploiting the polynomial
\[\Int:=D^{-1}=\left(\frac{1}{2n}\right)\left(-S+S^{-1}\right).\]
These algorithms compute the polynomial~$\Int^{k}\varphi(L)$, that is a numerator of~$\varphi(L)=\Int^{-k}\Int^k\varphi(L)$. Thus, by Theorem~\ref{main}, their result is a recurrence operator annihilating the coefficients of Chebyshev series solutions of~$L$. 

If~$p_k(1)p_k(-1)\neq0$, Proposition~\ref{prop:Lewanowicz} shows that~$\Int^k$ is the denominator of the irreducible fraction and therefore in that case all algorithms compute the irreducible fraction.
Otherwise, the result of these algorithms may have larger order than that returned by Lewanowicz' algorithm.
\begin{example} The function $(1-x^2)^{-1/4}$ has been dealt with in Example~\ref{ex-5}.
Lewanowicz' algorithm returns the second order recurrence~\eqref{eq coeffs ex-5}.
The numerator returned by the other algorithms has order~4:
\[(2n+1)c_n-4(n+2)c_{n+2}+(2n+7)c_{n+4}=0.\]
It is however possible to recover the smaller order recurrence: the gcld of $A=(2n+7)S^4-4(n+2)S^2+(2n+1)$ with $\Int$ is $\Int$, so that $A$ factors as $A=(S^2-1)P$ with~$P$ as in Example~\ref{ex-5}.
\end{example} 
More generally, dividing out the result of the computation of $I^k\phi(L)$ on the left by the gcld with~$\Int^k$ yields the result of Lewanowicz's algorithm.
\subsection{Paszkowski's Algorithm}
The starting point of Paszkowski's algorithm is to rewrite~$L$ from~\eqref{eq:L} as
\[L=\sum_{i=0}^k\partial_x^iq_i(x).\]
The polynomials~$q_i$ can be computed inductively starting with~$q_k=p_k$ and subtracting~$\partial_x^kq_k$ to produce a smaller order operator.
Then
\begin{equation}\label{Paszkowski}
\Int^k\varphi(L)=\sum_{i=0}^k\Int^{k-i}q_i(X).
\end{equation}
Algorithm~\ref{alg P1} follows.
\begin{algorithm}
  \caption{Paszkowski's Algorithm}
  \label{alg P1}
  \begin{algorithmic}
    \REQUIRE $L=\sum_{i=0}^k p_i(x)\partial_x^i$
    \ENSURE $\Int^k\varphi(L)$
	\STATE Compute $q_i$'s such that $L=\sum_{i=0}^k\partial_x^iq_i(x)$
	\STATE $R:=q_k(X)$
	\FORALL{$i$ from $1$ to $k$ }
    	\STATE $R:=R+\Int^{i}q_{k-i}(X)$
	\ENDFOR
 	\RETURN $R$
  \end{algorithmic}
\end{algorithm}

\subsection{Rebillard's Algorithm}                        
The starting point of Rebillard's algorithm is the identity
\[X_k=\Int^kXD^k=(2n)^{-1}((n+k)S+(n-k)S^{-1}),\]
that follows from an easy induction.
From there, he deduces
\[\Int^k\varphi(L)=\sum_{i=0}^k{\Int^kp_i(X)D^k\Int^{k-i}}=\sum_{i=0}^{k}{p_i(X_k)\Int^{k-i}}.\]
Algorithm~\ref{alg Reb} follows.
\begin{algorithm}
  \caption{Rebillard's Algorithm}
  \label{alg Reb}
  \begin{algorithmic}
    \REQUIRE $L=\sum_{i=0}^k p_i(x)\partial_x^i$
    \ENSURE $\Int^k\varphi(L)$
	\STATE Compute $p_i(X_k)$, $i=0,\dots,k$
	\STATE $R:=p_k(X_k)$
	\FORALL{$i$ from $1$ to $k$ }
    	\STATE $R:=R+p_{k-i}(X_k)\Int^{i}$
	\ENDFOR
	\RETURN $R$
  \end{algorithmic}
\end{algorithm}

\subsection{Complexity Analysis}
We now give a complexity analysis of Paszkowski's, Rebillard's and Lewanowicz' algorithms. This reveals a source of inefficiency for large orders, that we correct in our new algorithm in the next section.

We need to consider the size of  polynomials in two variables~$n$ and~$S$. We say that a polynomial has \emph{bidegree} $(m,p)$ in $(n,S)$ when it has degree~$m$ in~$n$ and $p$ in~$S$.

First, we state more precisely the shape of~$\Int^i$.
\begin{prop}[Rebillard \cite{Rebillard1998}]
	\label{Di}
	For all $i \in \N^*$,
\[\Int^{i}=\frac{1}{r(i)}\Bigg((n+1)_{i-1}S^{-i}
		+\sum_{k=1}^{i-1}{s(k) S^{-i+2k}}+(n-i+1)_{i-1}S^i\Bigg),\]
where $r(i)=2^i n\prod_{k=1}^{i-1}(n^2-k^2)$,
\[s(k)=(-1)^k\binom{i}{k}(n-i+2k)(n+k+1)_{i-1-k}(n-i+1)_{k-1},\]
and we use the Pochhammer symbol $(a)_i=a(a+1)\dotsm(a+i-1)$.
\end{prop}
In particular, the bidegree of~$r(i)\Int^i$ in $(n,S)$ is~$(i-1,2i)$. The proof is a tedious but easy induction that we omit here. From this formula follows a precise estimate of the size of the polynomial we are computing.
\begin{corollary}\label{coro:size}If $L$ in~\eqref{eq:L} has bidegree~$(d,k)$ in~$(x,\partial_x)$, then $r(k)\Int^k\varphi(L)$ is a polynomial of bidegree  in~$(n,S)$ at most $(2k-1,2(k+d))$.
\end{corollary}
\begin{proof}
First, $L$ can be rewritten as in Paszkowski's algorithm $\sum_{i=0}^k{\partial_x^iq_i(x)}$ with $\deg q_i\le d$. The identity
\[r(k)\Int^k\varphi(L)=\sum_{i=0}^k{\frac{r(k)}{r(i)}(r(i)I^i)q_{k-i}(X)}\]
shows that this is a polynomial in~$n$. Each term of the sum is the product of a polynomial of bidegree~$(2(k-i),0)$, a polynomial of bidegree $(i-1,2i)$, a polynomial of bidegree at most~$(0,2d)$. Thus each summand has bidegree at most $(2k-i-1,2i+2d)$, whence the result.
\end{proof}
\begin{prop}\label{theoPasz}Given $L$ as above for input, Paszkowski's algorithm requires~$O(dk^3)$ arithmetic operations in~$\Q$.
\end{prop}
\begin{proof}
The first step is the computation of the~$q_i$'s from the~$p_i$'s. 
The inductive method requires only $O(dk^2)$ arithmetic operations. 
Using ideas from~\cite{BoChLR08}, it is actually possible to decrease this complexity further to $O(\M(dk))$ operations~\cite{Bo09} (here, $\M$ is the complexity of polynomial product, see, e.g., \cite{GathenGerhard1999}).

The next step is the loop. The main cost in step $i$ is the multiplication of $\Int^{i}$ by $q_{k-i}(X)$. We multiply a polynomial of bidegree $(i-1,2i)$, with a polynomial in $S$ only, of degree~$2d$. The cost of this multiplication is $O(i^2d)$ arithmetic operations.
Summing for $i$ up to~$k$ gives the result.
\end{proof}

\begin{prop}
In the same conditions,	Rebillard's algorithm requires $O(d^3k+d^2k^3)$ arithmetic operations.
\end{prop}
\begin{proof}
The first step is the computation of the $p_i(X_k)$. The polynomial $X_k^i$ has bidegree $(3i,2i)$ in $(n,S)$. Then each $p_i(X_k)$ can be computed in $O(d^3)$ operations and all of them in $O(d^3k)$ operations.

The cost of the $i$th step of the loop is dominated by the cost of the multiplication of $p_{k-i}(X_k)$ by $\Int^i$.
The polynomial $p_{k-i}(X^k)$ has bidegree $(3d,2d)$ in $(n,S)$, while $\Int^i$ has bidegree $(2i-1,2i)$. Naive multiplication then requires $O(d^2i^2)$ operations. Summing over $k$ gives the result.
\end{proof} 

The output of Lewanowicz' algorithm is different in general. We give a comparison in the cases when it coincides.
\begin{prop}In the same conditions, and if all the gcrd during its execution are trivial, Lewanowicz' algorithm requires 
$O(dk^3)$ arithmetic operations.
\end{prop}
\begin{proof}
We only give a sketch. When all gcrd's are trivial, it turns out that the computation of lclm's and cofactors is of the same order of complexity as the computation of the product~$Q p_i$, where moreover~$Q=I^{k-i}$. This is the same as in the analysis of Paszkowski's algorithm.
\end{proof}

\subsection{New Fast Algorithm}
We now give another algorithm for the same operator $\Int^k\varphi(L)$. The design of our algorithm is motivated by computational complexity issues. In the  analyses above, most of the complexity comes from the fact that during the computations, the bidegrees of the intermediate polynomials grow linearly and they are multiplied by polynomials of fixed degree. Instead, we aim at balancing degrees so as to make use of the recent fast algorithm for the product of linear differential operators~\cite{FFTjoris,BoChLR08}, that we denote FFT-mult. We achieve the following complexity.

\begin{theo}\label{complexity}
Algorithm \ref{alg new} computes the recurrence operator~$\Int^k\varphi(L)$ in $O((d+k)k^{\omega-1})$ arithmetic operations. 
\end{theo}
Here, $\omega$ is a feasible exponent for matrix multiplication with coefficients in $\Q$ (see, e.g., \cite{GathenGerhard1999}). We now prove this result.
\begin{algorithm}
	\caption{Divide and Conquer}
	\label{alg new}
	\begin{algorithmic}
		\REQUIRE Polynomials $a_0(x),\dots,a_k(x)$
		\ENSURE $P_{(0,\dots,k)}=\sum_{i=0}^{k}\Int^{i}a_i(X)$
		\IF {$k=0$} \RETURN $a_0(X)$
		\ELSE
			\STATE $\ell:=\lceil k/2\rceil$
			\STATE Compute recursively $P_{(0,\dots,\ell-1)}$ and $P_{(\ell,\dots,k)}$.
			\RETURN $P_{(0,\dots,\ell-1)}+\Int^\ell P_{(\ell,\dots,k)}$.
		\ENDIF
	\end{algorithmic}
\end{algorithm}  

Starting from~\eqref{Paszkowski}, we write
\[
\sum_{i=0}^{k}{\Int^iq_{k-i}(X)}=   
\sum_{i=0}^{\ell-1}{\Int^iq_{k-i}(X)} +
\Int^\ell\sum_{i=\ell}^{k}{\Int^{i-\ell}q_{k-i}(X)}
=:P_{(0,\cdots,l-1)} + I^lP_{(l,\cdots,k)}.
\]
We choose~$\ell=\lceil k/2\rceil$ and apply the same idea recursively.

The time consuming part of the computation is the product $\Int^\ell P_{(\ell,\dots,k)}$, for which we give a specialized algorithm.
\begin{algorithm}
	\caption{Fast Multiplication\label{fast-mult}}
	\label{mult}
	\begin{algorithmic}
		\REQUIRE $I^{\ell}$ and $P:=\sum_{i=0}^\ell I^ia_i(X)$
		\ENSURE $I^\ell P$
		\STATE Decompose $P$ as in Eq.~\eqref{eq:decomp}
		\STATE $R:=0$
		\FORALL{$i$ from 0 to $\lfloor (\ell+d)/\ell\rfloor$}
			\STATE $R:=R+\operatorname{FFT-mult}(r(\ell)\Int^\ell,A_i)S^{-d-\ell+i(k+1)}$
		\ENDFOR
		\RETURN $1/r(2l)R$.
	\end{algorithmic}
\end{algorithm}

To simplify the presentation, assume~$k=2\ell$. Corollary~\ref{coro:size} implies that~$\Int^\ell$ has degree~$2\ell$ in~$S$, $P_{(\ell,\dots,k)}$ has degree at most $2\ell+2d$ in $S$. They have rational function coefficients whose degrees are also bounded by this result. If $d$ is large, the degrees in~$S$ are unbalanced, so we first decompose
\begin{equation}\label{eq:decomp}
P_{(\ell,\dots,k)}=A_0(n,S)S^{-d-\ell}+A_1(n,S)S^{-d+\ell+1}+\dotsb,
\end{equation}
where the~$A_i$'s have degree at most~$2\ell$ in~$S$. Note that this decomposition is only an extraction of coefficients and does not use any arithmetic operation.

We are thus left with the product of~$\Int^\ell$ with the~$A_i$'s. Although both have rational function coefficients and thus cannot be multiplied directly by FFT-mult, we also have that $r(\ell)\Int^\ell P_{(\ell,\dots,k)}$ has polynomial coefficients in $n$ of degree at most~$k-1$ and therefore so does $r(\ell)\Int^\ell A_i$. To perform the product efficiently, we make use of the fact that FFT-mult proceeds by evaluation and interpolation: during the evaluation phase, we evaluate the \emph{rational function} coefficients of~$A_i$ as if they were polynomials (and within the same complexity thanks to our degree bounds), avoiding the zeros~$-\ell,\dots,\ell$ of their denominators; similarly, we evaluate the polynomial coefficients of~$r(\ell)\Int^\ell$. Then we compute the necessary products. With the bounds on the degree in~$n$ we have for the \emph{polynomial} coefficients in the result, the interpolation phase then returns the result. The complexity of each of these multiplications is thus $O(\ell^\omega)$ operations. The algorithm for this multiplication is summarized in Algorithm~\ref{fast-mult}. Note also that a constant factor can be saved by not recomputing the ``FFT'' of $r(\ell)\Int^\ell$ at each time.

\begin{prop}
The cost of multiplying $I^{\ell}$ by $\sum_{i=0}^\ell I^{i}a_i(X)$ with~$\deg a_i\le d$ using Algorithm~\ref{fast-mult} is $O((\ell+d)\ell^{\omega-1})$ arithmetic operations.
\end{prop}

\begin{proof}
We have seen that each multiplication~$r(\ell)\Int^\ell A_i$ has complexity $O(\ell^\omega)$. This is performed $(\ell+d)/\ell$ times. Right multiplication by powers of~$S$ does not use any arithmetic operations. The additions require a smaller number of operations, whence the result.
\end{proof}

Now, let $T(k,d)$ be the complexity of Algorithm~\ref{alg new}. Using this proposition, we get
\[T(k,d)=2T(k/2,d)+O((d+k)k^{\omega-1}),\] 
the complexity estimate in Theorem~\ref{complexity} follows from the convergence of the geometric series. (Again, a constant factor can be saved by computing the powers of~$I$ only once.)

\subsection{Algorithm by Rebillard and Zakraj\v{s}ek}\label{rebillard-zakrajsek}
In~\cite{RebillardZakrajsek2006}, an algorithm of a different nature is proposed. It does not compute a numerator of~$\phi(L)$, but manages in some cases to derive a smaller order recurrence corresponding to a right factor of the numerator of~$\phi(L)$. We plan to come back to this algorithm in connexion to minimality issues in future work. Here, we merely give a few indications and comments on special cases.
\begin{example} The following is taken from~\cite{RebillardZakrajsek2006}. Starting from the differential operator $L=(x+1)^2\partial_x^2-(x+1)\partial_x+x+7/4$, the computation of~$\phi(L)$ by Lewanowicz's algorithm leads to a numerator of order~4, whereas the algorithm in~\cite{RebillardZakrajsek2006} produces one of order only~3. We note that this operator can also be obtained by Lewanowicz's algorithm, applied to~$\partial_x L$ instead of~$L$. In many cases, this technique applies.
\end{example}
Since the algorithm in~\cite{RebillardZakrajsek2006} computes a right factor of the numerator of~$\phi(L)$, analytic hypotheses such as (H) or (H') in our Theorem~\ref{main} are necessary.
\begin{example} In the case of $\arccos(x)$, the numerator of $\phi(L)$ is a constant (see Example~\ref{ex arccos}), so that this is also the result of the algorithm in~\cite{RebillardZakrajsek2006}. Starting from $(1-x^2)L$, we obtained an operator of order~4 in Ex.~\ref{ex arccos}. On this operator, the algorithm in~\cite{RebillardZakrajsek2006} returns an operator of order~2. 
\end{example}

\section{Examples}\label{sec:examples}
The fast algorithm does not lend itself easily to an efficient implementation in Maple, since it relies on fast evaluation/interpolation and fast matrix product. We have however implemented the slow algorithms in Maple and show how other algorithms from computer algebra can sometimes be applied to the resulting recurrences, so that nice expression for the coefficients can be recovered. We have also implemented variants of Horner-like evaluations that seem to perform well, see~\cite{Benoit2008}.

\begin{example}[arctan] Starting from $(x^2+1)\partial_x^2+2x\partial_x$, we get
	\[nc_{{n}}+ ( 6n+12) c_{{n+2}}+ ( n+4 ) c_{{n+4}}=0.\]
The initial conditions are computed by Maple as $c_0=c_2=0$ and $c_1=2\sqrt{2}-2$, $c_3=(14-10\sqrt{2})/3$. The recurrence can then be solved by Petkov\v sek's algorithm and we get
\[\arctan(x)=2\sum_{k\ge0}{(1-\sqrt{2})^{2k+1}\frac{T_{2k+1}(x)}{2k+1}}.\]
\end{example}

\begin{example}[error function] Starting from $\partial_x^2+2x\partial_x$, we get a more complicated recurrence:
	\[( {n}^{2}+3n ) c_{{n}}+ ( 2{n}^{3}+12{n}^{2}+24n+16 ) c_{{n+2}}- ( {n}^{2}+5n+4 ) c_{{n+4}}=0.\]
A closed form is known to be
\[c_{2k}=0,\quad c_{2k+1}=\frac{(-1)^k}{\sqrt{e\pi}}\frac{I_k(1)+I_{k+1}(1)}{2k+1},\]
but it seems that the algorithms in computer algebra are not strong enough to find this automatically, yet.
\end{example}

\begin{example}[arctanh] Starting from $L=(x^2-1)\partial_x^2+2x\partial_x$, we get \[(n+2)c_{n+2}-nc_n=0\]
by computing the numerator of~$\phi(L)$. Although neither of our hypotheses (H) or (H') holds here, this result is correct, as can be checked from the expansion
\[\operatorname{arctanh}(x)=2\sum_{k\ge0}\frac{T_{2k+1}(x)}{2k+1}.\]
Again, this suggests that more work on obtaining a recurrence of minimal order is necessary.
\end{example}

%
%
%



\end{document}